\documentclass[12 pt]{article}

\pdfoutput=1

\usepackage{fullpage}
\usepackage{amscd}
\usepackage{amsmath}
\usepackage{amssymb}
\usepackage{amsthm}
\usepackage{graphicx}
\usepackage{epstopdf}
\usepackage{latexsym}
\usepackage{verbatim}
\input epsf.tex

\newtheorem{theorem}{Theorem}[section]

\newtheorem{corollary}[theorem]{Corollary}

\newtheorem{proposition}[theorem]{Proposition}

\theoremstyle{definition}
\newtheorem{assumption}[theorem]{Assumption}

\newtheorem{definition}[theorem]{Definition}

\theoremstyle{remark}
\newtheorem{example}[theorem]{Example}

\newcounter{tempthm}
\newcounter{tempsec}

\newcommand{\savecounter}[1]{\newcounter{thmcounter#1}
\setcounter{thmcounter#1}{\value{theorem}}
\newcounter{seccounter#1}
\setcounter{seccounter#1}{\value{section}}}

\newcommand{\usesavedcounter}[1]{\setcounter{tempthm}{\value{theorem}}
\setcounter{theorem}{\value{thmcounter#1}}
\setcounter{tempsec}{\value{section}}
\setcounter{section}{\value{seccounter#1}}}

\newcommand{\restorecounter}{\setcounter{theorem}{\value{tempthm}}
\setcounter{section}{\value{tempsec}}}

\newcommand{\Rm}{\mathbb{R}}

\newcommand{\abskap}{\lvert\kappa\rvert}
\newcommand{\abskapi}{\lvert\kappa_i\rvert}

\newcommand{\shortv}[1]{}

\DeclareMathOperator{\supp}{supp}
\DeclareMathOperator{\sign}{sign}

\title{\LARGE \bf
Structure of Extreme Correlated Equilibria: \\ a Zero-Sum Example and its Implications
}

\author{Noah D. Stein, Asuman Ozdaglar, and Pablo A. Parrilo
\thanks{Department of Electrical Engineering,
        Massachusetts Institute of Technology: Cambridge, MA 02139.
        {\tt\small nstein@mit.edu}, {\tt\small asuman@mit.edu}, and {\tt\small parrilo@mit.edu}.}
\thanks{This research was funded in part by National Science Foundation grants $\text{DMI-}0545910$ and ECCS-$0621922$ and AFOSR MURI subaward $2003\text{-}07688\text{-}1$.}
      }

\begin{document}

\markright{LIDS Technical Report $2829$}

\maketitle

\thispagestyle{headings}

\pagestyle{plain}

\begin{abstract}
We exhibit the rich structure of the set of correlated equilibria by analyzing the simplest of polynomial games: the mixed extension of matching pennies.  We show that while the correlated equilibrium set is convex and compact, the structure of its extreme points can be quite complicated.  In finite games the ratio of extreme correlated to extreme Nash equilibria can be greater than exponential in the size of the strategy spaces.  In polynomial games there can exist extreme correlated equilibria which are not finitely supported; we construct a large family of examples using techniques from ergodic theory.  We show that in general the set of correlated equilibrium distributions of a polynomial game cannot be described by conditions on finitely many moments (means, covariances, etc.), in marked contrast to the set of Nash equilibria which is always expressible in terms of finitely many moments.
\end{abstract}

\section{Introduction}
Correlated equilibria are a natural generalization of Nash equilibria introduced by Aumann \cite{a:scrs}.  They are defined to be joint probability distributions over the players' strategy spaces, such that if each player receives a private recommendation sampled according to the distribution, no player has an incentive to deviate unilaterally from playing his recommended strategy.  In finite games the set of correlated equilibria is a compact convex polytope, and therefore seemingly much simpler than the set of Nash equilibria, which can be essentially any algebraic variety \cite{d:une}.  Even in the simple case of two-player finite games, the set of Nash equilibria is a union of finitely many polytopes: seemingly more complicated than the set of correlated equilibria.

Nonetheless we will see that there are two-player zero-sum games in which the set of correlated equilibria has many more extreme points than the set of Nash equilibria has.  This behavior does not seem to be pathological in any way: it occurs in very simple finite games and the simplest of infinite games.  We take this as evidence that this complexity is likely to be quite common.

\paragraph{Contributions}
\begin{itemize}
\item We give a family of examples of two-player zero-sum finite games in which the set of Nash equilibria has polynomially many extreme points  (Section~\ref{sec:nash}), while the set of correlated equilibria has factorially many extreme points (Section~\ref{sec:corr}).

For bimatrix games, this shows that while extreme Nash equilibria are a subset of the extreme correlated equilibria (see Related Literature below), enumerating all the extreme correlated equilibria is in general a bad way of computing all the extreme Nash equilibria.  In particular, it would be faster to enumerate all subsets of the strategy spaces (there are ``only'' exponentially many) and check whether each was the support of a Nash equilibrium.

\item We give a related example of a continuous game with strategy sets equal to $[-1,1]$ and bilinear utility functions.  This game is just the mixed extension of matching pennies, but we show that it has extreme correlated equilibria with arbitrarily large finite support (Proposition~\ref{prop:mpextfinsuppcorr}) and also with infinite support (Proposition~\ref{prop:mpextcorr}).  This is in contrast to the extreme Nash equilibria, which always have uniformly bounded finite support in zero-sum games with polynomial utilities \cite{karlin:tig}.

Once the existence of Nash equilibria in continuous games has been established \cite{glicksberg:cg}, it is straightforward to show that polynomial games\footnote{For the purposes of this paper a \emph{polynomial game} is one in which the strategy spaces are compact intervals and the utility functions are polynomials in the players' strategies.} admit Nash equilibria with finite support \cite{karlin:tig, sop:slrcg}.  There are more elementary ways of showing that continuous games have correlated equilibria \cite{hs:ece}, but to the authors' knowledge there is no proof that polynomial games have finitely supported correlated equilibria which does not rely on the existence of Nash equilibria.  This example shows that the plausible-sounding proof idea that all extreme correlated equilibria are finitely supported simply isn't true.

\item Comparing Proposition~\ref{prop:momentmeasure} with this example shows that in general there is no finite-dimensional description of the set of correlated equilibria of a zero-sum polynomial game.  That is to say, one cannot check if a measure is a correlated equilibrium merely by examining finitely many generalized moments (parameters such as mean, covariance, etc.\ -- any compactly supported distribution can be specified by countably many such parameters).  Such a description for the Nash equilibria has been known for over fifty years \cite{karlin:tig}.

Intuitively, the reason for this difference is that being a correlated equilibrium is a statement about conditional distributions, and these are too delicate to be controlled by finitely many moments.  This example confirms the intuition.

Experience from finite games suggests that correlated equilibria should be easier to compute than Nash equilibria.  While there are computational methods which converge asymptotically to correlated equilibria of polynomial games \cite{spo:cecgcc}, the only exact algorithm the authors are aware of consists of computing a Nash equilibrium by quantifier elimination (extremely slow), which is possible because of the finite-dimensional description.  In particular, no provably efficient method for computing correlated equilibria of polynomial games exactly or approximately is known.  The lack of a finite-dimensional description of the problem seems to be an important part of what makes it difficult.
\end{itemize}



\paragraph{Related Literature}
The geometry of Nash and correlated equilibria has been studied extensively.  Therefore we only mention work below if it is directly connected to ours and we do not attempt to be exhaustive.

The result most closely related to the present paper states that in two-player finite games, extreme Nash equilibria (viewed as product distributions) are a subset of the extreme correlated equilibria.  Cripps \cite{c:ecnetpg} and Evangelista and Raghavan \cite{er:nce} proved this independently.  This result shows that it makes sense to compare the number of extreme Nash and correlated equilibria.  It also raises the natural question of whether all extreme Nash equilibria could be enumerated efficiently by enumerating the extreme correlated equilibria.  We show that there can be many more extreme correlated equilibria than extreme Nash equilibria, answering this question in the negative.

In a similar vein, Nau et al.\ \cite{nch:ognece} show that for non-trivial finite games with any number of players, the Nash equilibria lie on the boundary of the correlated equilibrium polytope.  With three or more players, the Nash equilibria need not be extreme correlated equilibria.  For example consider the three-player poker game analyzed by Nash in \cite{nash:ncg} which has rational payoffs, hence rational extreme correlated equilibria, but whose unique Nash equilibrium uses irrational probabilities.

Separable games, a generalization of polynomial games, were first studied around the $1950$'s by Dresher, Karlin, and Shapley in papers such as \cite{dks:pg}, \cite{dk:scgfp}, and \cite{ks:gms}, which were later combined in Karlin's book \cite{karlin:tig}.  Their work focuses on the zero-sum case, which contains some of the key ideas for the nonzero-sum case.  In particular, they show how to replace the infinite-dimensional mixed strategy spaces (sets of probability distributions over compact metric spaces) with finite-dimensional moment spaces.  Carath\'{e}odory's theorem \cite{bno:convex} then applies to show that finitely-supported Nash equilibria exist.

There are many similarities between separable games and finite games whose payoff matrices satisfy low-rank conditions.  Lipton et al.\ \cite{lmm:plgss} consider two-player finite games and provide bounds on the cardinality of the support of extreme Nash equilibrium strategies in terms of the ranks of the payoff matrices.  The main technical tool here is again Carath\'{e}odory's theorem.

Germano and Lugosi show that in finite games with three or more players there exist correlated equilibria with smaller support than one might expect for Nash equilibria \cite{gl:essce}.  The proof is geometrical; it essentially views correlated equilibria as living in a subspace of low codimension and it too uses Carath\'{e}odory's theorem \cite{bno:convex}.  

The bounds on the support of equilibria in finite and separable games of the previous three paragraphs are all synthesized in \cite{s:mastersthesis}; the portion on Nash equilibria has appeared in \cite{sop:slrcg}.  The general idea is that simple payoffs (low-rank matrices, low-degree polynomials, etc.) lead to simple Nash equilibria (small support), and those in turn lead to simple correlated equilibria (small support again).

To produce upper bounds on the minimal support of correlated equilibria which depend only on the rank of the payoff matrices and not on the size of the strategy sets, this work does not bound the support of all extreme correlated equilibria, but rather only those whose support is contained inside a Nash equilibrium of small support, which must exist.  Similar results hold for polynomial games with, for example, degree used in place of rank (the notions of degree and rank are generalized in \cite{s:mastersthesis} and \cite{sop:slrcg}).

This work left open the question of whether all extreme correlated equilibria have support size which can be bounded in terms of the rank of the payoff matrices, independently of the size of the strategy sets.  Here we show that this is not the case, because our examples have payoffs which are of rank $1$ and extreme correlated equilibria of arbitrarily large, even infinite, support.

Correlated equilibria without finite support have been defined and studied by several authors.  An important example of this line of research is the paper by Hart and Schmeidler \cite{hs:ece}.  The definition of correlated equilibria presented in \cite{hs:ece} is convenient for proving some theoretical results (they focus on existence) but not usually for computation.

The authors of the present paper have developed several equivalent characterizations of correlated equilibria in continuous games which are more suitable for computation \cite{spo:cecgcc}.  One of these forms the basis for the analysis in Section~\ref{sec:corr} below.  Other such characterizations lead to algorithms for approximating correlated equilibria of continuous games \cite{spo:cecgcc}.

\paragraph{Outline}
The remainder of this paper is organized as follows.  Section $2$ introduces the examples to be studied.  The two types of example are closely related -- the finite game examples are just restrictions of the strategy spaces in the infinite game example to fixed finite sets.  This allows us to analyze both examples on equal footing.  In Section $3$ we define and compute the extreme Nash equilibria of these examples, counting them in the finite game example.  Then we define and analyze the extreme correlated equilibria in Section $4$.  This analysis is somewhat long and at times technical, so we present a detailed roadmap before beginning.  We close with Section $5$, where we outline directions for future work.

\section{Description of the examples}
\label{sec:examples}

First we fix notation.  When $S$ is a topological space, $\Delta(S)$ will denote the set of Borel probability measures on $S$ and $\Delta^*(S)$ the set of finite Borel measures on $S$.  In particular $\Delta(S)$ is the set of measures in $\Delta^*(S)$ with unit mass.  If $S$ is finite it will be given the discrete topology by default so $\Delta(S)$ is a simplex and $\Delta^*(S)$ is an orthant in $\Rm^{\lvert S \rvert}$.  We abuse notation and write the measure of a singleton $\{p\}$ as $\mu(p)$ rather than $\mu(\{p\})$.  For any $p\in S$, define $\delta_p\in \Delta(S)$ to be the measure which assigns unit mass to the point $p$.  Let $I = [-1,1]\subset\Rm$.

We will focus on two related examples, one with finite strategy sets and one with infinite strategy sets.  We will develop them in parallel by analyzing arbitrary games satisfying the following condition.  The condition does not have any game theoretic content; it was merely chosen for simplicity and the results to which it leads.

\begin{assumption}
\label{assump:mp}
The game is a zero-sum strategic form game with two players, called $X$ and $Y$.  The strategy sets $C_X$ and $C_Y$ are compact subsets of $I = [-1,1]$, each of which contains at least one positive element and at least one negative element.  Player $X$ chooses a strategy $x\in C_X$ and player $Y$ chooses $y\in C_Y$.  The utility functions\footnote{By inspection of the utilities we can see that for any $C_X$ and $C_Y$ with at least two points, the rank of this game in the sense of \cite{sop:slrcg} is $(1,1)$ (and in fact also in the stronger sense of Theorem $3.3$ of that paper).  The notion of the rank of a game is related to the rank of the payoff matrices and will not play a significant role in this paper; we merely wish to note that under this definition of complexity of payoffs, the games we consider are extremely simple.} are $u_X(x,y) = xy = -u_Y(x,y)$.
\end{assumption}

\savecounter{exfin}
\begin{example}
\label{ex:finite}
Fix an integer $n>0$.  Let $C_X$ and $C_Y$ each have $2n$ elements, $n$ of which are positive and $n$ of which are negative.  If we take $n = 1$ and $C_X = C_Y = \{-1,1\}$ then we recover the matching pennies game, as shown in Table~\ref{tab:matchingpennies}.
\end{example}

\begin{table}
\centering
\begin{tabular}{|c|cc|}
\hline
$(u_X,u_Y)$ & $x = -1$ & $x = 1$ \\
\hline
$y = -1$ & $(1,-1)$ & $(-1,1)$\\
$y = 1$ & $(-1,1)$ & $(1,-1)$\\
\hline
\end{tabular}
\caption{Utilities for matching pennies}
\label{tab:matchingpennies}
\end{table}

\savecounter{exinf}
\begin{example}
\label{ex:infinite}
Let $C_X = C_Y = [-1,1]$.  Then the game is essentially the mixed extension of matching pennies.  That is to say, suppose two players play matching pennies and choose their strategies independently, playing $1$ with probabilities $p\in [0,1]$ and $q\in [0,1]$.  Define the utilities for the mixed extension to be the expected utilities under this random choice of strategies.  Letting $x = 2p - 1$ and $y = 2q - 1$, the utility to the first player is $xy$ and the utility to the second player is $-xy$.  Therefore this example is the mixed extension of matching pennies, up to an affine scaling of the strategies.
\end{example}

Usually one looks at pure equilibria of the mixed extension of a game; these are exactly the mixed equilibria of the original game.  We will instead be looking at mixed Nash equilibria and correlated equilibria of the mixed extension itself, a game with a continuum of actions.  The relationship between correlated equilibria of the mixed extension and those of the original game is much more complicated than the corresponding relationship for mixed Nash equilibria.  This drives the results of the paper.

\section{Extreme Nash equilibria}
\label{sec:nash}
We now characterize and count the extreme points of the sets of Nash equilibria in games satisfying Assumption~\ref{assump:mp}.  Since the games are zero-sum, the set of Nash equilibria can be viewed as a Cartesian product of two (weak*) compact convex sets, the sets of maximin and minimax strategies \cite{glicksberg:cg}.  The Krein-Milman theorem completely characterizes such sets by their extreme points \cite{r:fa}, explaining our focus on extreme points throughout.  

We define Nash equilibria in two-player games, which will be sufficient for our purposes, as well as the standard notions of extreme point and extreme ray from convex analysis.  

\begin{definition}A \textbf{Nash equilibrium} is a pair $(\sigma,\tau)\in \Delta(C_X)\times\Delta(C_Y)$ such that $u_X(x,\tau)\leq u_X(\sigma,\tau)$ for all $x\in C_X$ and $u_Y(\sigma,y)\leq u_Y(\sigma,\tau)$ for all $y\in C_Y$ (where we extend utilities by expectation in the usual fashion $u_X(x,\tau) = \int u_X(x,y)\,d\tau(y)$, etc.).
\end{definition}

In other words, a Nash equilibrium is a strategy pair in which each player is playing a best reply to his opponent's strategy.

\begin{definition}
A point $x$ in a (usually convex) subset $K$ of a real vector space is an \textbf{extreme point} if $x = \lambda y + (1-\lambda)z$ for $y,z\in K$ and $\lambda\in (0,1)$ implies $x = y = z$.
\end{definition}

The related notion of extreme ray will not be used until the next section, but we record it here for comparison.

\begin{definition}
A convex set $K$ such that $x\in K$ and $\lambda\geq 0$ implies $\lambda x \in K$ is called a \textbf{convex cone}.  A point $x\neq 0$ is an \textbf{extreme ray} of the convex cone $K$ if $x = y + z$ and $y,z\in K$ implies that $y$ or $z$ is a scalar multiple of $x$.
\end{definition}

The Nash equilibria of games satisfying Assumption~\ref{assump:mp} take the following particularly simple form.

\begin{proposition}
\label{prop:nashzeromean}
A pair $(\sigma,\tau)\in\Delta(C_X)\times\Delta(C_Y)$ is a Nash equilibrium of a game satisfying Assumption~\ref{assump:mp} if and only if $\int x\,d\sigma(x) = \int y\,d\tau(y) = 0$.
\end{proposition}

\begin{proof}
If $\int x\,d\sigma(x) = 0$ then $u_Y(\sigma, y) = 0$ for all $y\in C_Y$, so any $\tau\in\Delta(C_Y)$ is a best response to $\sigma$.  If $\int y\,d\tau(y) = 0$ as well then $\sigma$ is also a best response to $\tau$, so $(\sigma,\tau)$ is a Nash equilibrium.

Suppose for a contradiction that there exists a Nash equilibrium $(\sigma,\tau)$ such that $\int x\,d\sigma(x) > 0$; the other cases are similar.  Player $y$ must play a best response, so $\int y\,d\tau(y) < 0$, which is possible by assumption.  Player $x$ plays a best response to that, so $\int x\,d\sigma(x) < 0$, a contradiction.  
\end{proof}

We introduce the notion of extreme Nash equilibrium in two-player zero-sum games.  For an extension of this definition to two-player finite games and a proof that extreme Nash equilibria are always extreme points of the set of correlated equilibria in this setting, see \cite{c:ecnetpg} or \cite{er:nce}.

\begin{definition}
In a two-player zero-sum game, \textbf{maximin} and \textbf{minimax} strategies are those mixed strategies for player $X$ and $Y$, respectively, which appear in a Nash equilibrium.  A Nash equilibrium of a zero-sum game is called \textbf{extreme} if $\sigma$ and $\tau$ are extreme points of the maximin and minimax sets, respectively.
\end{definition}

Applying Proposition~\ref{prop:nashzeromean} to this definition, we can characterize the extreme Nash equilibria of games satisfying Assumption~\ref{assump:mp}.

\begin{proposition}
\label{prop:extremenash}
Consider a game satisfying Assumption~\ref{assump:mp}.  A pair $(\sigma,\tau)\in\Delta(C_X)\times\Delta(C_Y)$ is an extreme Nash equilibrium if and only if $\sigma$ and $\tau$ are each either $\delta_0$ or of the form $\alpha\delta_u + \beta\delta_v$ where $u<0$, $v, \alpha,\beta>0$, $\alpha+\beta = 1$, and $\alpha u + \beta v = 0$.
\end{proposition}

\begin{proof}
By Proposition~\ref{prop:nashzeromean} we must show that these distributions are the extreme points of the set of probability distributions having zero mean.  Since $\delta_0$ is an extreme point of the set of probability distributions, it must be an extreme point of the subset which has zero mean.  To see that $\alpha\delta_u + \beta\delta_v$ is also an extreme point, suppose we could write it as a convex combination of two other probability distributions with zero mean.  The condition that both be positive measures implies that both must be of the form $\alpha'\delta_u + \beta'\delta_v$.  But $\alpha$ and $\beta$ as specified above are the unique coefficients which make this be a probability measure with zero mean.  Therefore $\alpha' = \alpha$ and $\beta' = \beta$, so $\alpha\delta_u + \beta\delta_v$ cannot be written as a nontrivial convex combination of probability distributions with zero mean, i.e., it is an extreme point.

Suppose $\sigma$ were an extreme point which was not of one of these types.  Then $\sigma$ could not be supported on one or two points, so either $[0,1]$ or $[-1,0)$ could be partitioned into two sets of positive measure.  We will only treat the first case; the second is similar.  Let $[0,1] = A\cup B$ where $A\cap B = \emptyset$ and $\sigma(A),\sigma(B)>0$.  Since $\sigma$ has zero mean we must have $\sigma([-1,0)) > 0$ as well.

For a set $D$ we define the restriction measure $\sigma\vert_D$ by $\sigma\vert_D(C) = \sigma(D\cap C)$ for all $C$.  Then $\sigma = \sigma\vert_A + \sigma\vert_B + \sigma\vert_{[-1,0)}$.  Let $a = \int_A x\,d\sigma(x)$, $b = \int_B x\,d\sigma(x)$, and $c = \int_{[-1,0)}x\,d\sigma(x)$.  Since $\sigma([-1,0))>0$ and $x$ is less than zero everywhere on $[-1,0)$, we must have $c < 0$ and similarly $a,b\geq 0$.  By assumption $a + b + c = 0$.  Therefore we can write:
\begin{equation*}
\sigma = \left(\sigma\vert_A + \frac{a}{\lvert c \rvert} \sigma\vert_{[-1,0)}\right) + \left(\sigma\vert_B + \frac{b}{\lvert c \rvert}\sigma\vert_{[-1,0)}\right)
\end{equation*}

Being an extreme point of the set of probability measures with zero mean, $\sigma$ must be an extreme ray of the set of positive measures with first moment equal to zero.  But this means that we cannot write $\sigma = \sigma_1 + \sigma_2$ where the $\sigma_i$ are positive measures with zero first moment unless $\sigma_i$ is a multiple of $\sigma$.  Neither of the measures in parentheses above is a multiple of $\sigma$, so we have a contradiction.
\end{proof}

We illustrate this proposition on both examples introduced in Section~\ref{sec:examples}.

\usesavedcounter{exfin}
\begin{example}[cont'd]
In this case neither $C_X$ nor $C_Y$ contains zero, so the only extreme Nash equilibria are those in which $\sigma$ and $\tau$ are of the form $\alpha\delta_u + \beta\delta_v$ for $u<0$ and $v>0$.  For any choice of $u$ and $v$ simple algebra gives unique $\alpha$ and $\beta$ satisfying the conditions of Proposition~\ref{prop:extremenash}.  There are $n$ possible choices for each of $u$ and $v$ for each of the two players, so there are $n^4$ extreme Nash equilibria.
\end{example}
\restorecounter

\usesavedcounter{exinf}
\begin{example}[cont'd]
Since $C_X = C_Y = [-1,1]$, there are infinitely many extreme Nash equilibria in this case.  However, they are all finitely supported and the size of the support of each player's strategy is always either one or two.  Furthermore the condition that $(\sigma,\tau)$ be a Nash equilibrium is equivalent to both having zero mean.  This illustrates the general facts that in games with polynomial utility functions the Nash equilibrium conditions only involve finitely many moments of $\sigma$ and $\tau$ (in this case, only the mean) and the extreme Nash equilibria (when defined, say for zero-sum games) have uniformly bounded support \cite{karlin:tig}.
\end{example}
\restorecounter

\section{Extreme correlated equilibria}
\label{sec:corr}
In this section we will show that even in finite games, the number of extreme correlated equilibria can be much larger than the number of extreme Nash equilibria.  It makes sense to compare these because all extreme Nash equilibria of a two-player game, viewed as product distributions, are automatically extreme correlated equilibria \cite{c:ecnetpg,er:nce}.

In the case of polynomial games we will show that there can be extreme correlated equilibria with arbitrarily large finite support and with infinite support.  This implies that the set of correlated equilibria cannot be characterized in terms of finitely many joint moments.

\paragraph{Roadmap}The analysis proceeds in several steps which will be technical at times, so we start with an outline of what follows.  
\begin{itemize}
\item We begin by defining correlated equilibria in games satisfying Assumption~\ref{assump:mp} using a characterization from \cite{spo:cecgcc}.
\item Proposition~\ref{prop:mpcorreqchar} shows that this characterization can be simplified because of our choice of utility functions.
\item We use this characterization to construct a family of finitely supported extreme correlated equilibria in Proposition~\ref{prop:mpextfinsuppcorr}.
\item Then we note that all extreme correlated equilibria of the games in Example~\ref{ex:finite} are of this form, so this allows us to count the extreme correlated equilibria and determine their asymptotic rate of growth as the number of pure strategies grows.
\item Next we introduce some ideas from ergodic theory.  With these tools in hand, we construct in Proposition~\ref{prop:mpextcorr} a large family of extreme correlated equilibria without finite support for the game in Example~\ref{ex:infinite}.
\item Finally we show that if a set can be represented by finitely many moments then all its extreme points have uniformly bounded finite support.  This shows that the set of correlated equilibria of the game in Example~\ref{ex:infinite} cannot be represented by finitely many moments and completes the analysis.
\end{itemize}

Having completed the roadmap, we are ready to begin.  Correlated equilibria are meant to capture the notion of a joint distribution of private recommendations to the two players such that neither player can expect to improve his payoff by deviating unilaterally from his recommendation.  For finitely supported probability distributions and games satisfying Assumption~\ref{assump:mp}, this can be written as per the standard definition (see \cite{myerson:gt} or \cite{ft:gt}):

\begin{definition}
\label{def:ce}A finitely supported probability distribution $\mu\in\Delta(C_X\times C_Y)$ is a \textbf{correlated equilibrium} of a game satisfying Assumption~\ref{assump:mp} if
\[
\sum_{y\in C_Y} \mu(x,y)[xy - x'y]\geq 0
\]
for all $x,x'\in C_X$ and 
\[
\sum_{x\in C_X}\mu(x,y)[xy'-xy]\geq 0
\]
for all $y,y'\in C_Y$ (note that these sums are finite by the assumption on $\mu$).
\end{definition}
The standard definition extending this notion to arbitrary (not necessarily finitely supported) distributions is given in \cite{hs:ece}.  This definition is difficult to compute with, so we will use the following equivalent characterization.

\begin{proposition}[\cite{spo:cecgcc}]
\label{prop:correqchar1}A probability distribution $\mu\in\Delta(C_X\times C_Y)$ is a correlated equilibrium of a game satisfying Assumption~\ref{assump:mp} if and only if
\begin{equation*}
\int_{A\times I} (xy-x'y)\,d\mu(x,y)\geq 0\text{\ \ \ and\ \ \ }\int_{I\times A}(xy-xy')\,d\mu(x,y)\leq 0
\end{equation*}
for all $x'\in C_X$, $y'\in C_Y$, and measurable $A\subseteq I$.
\end{proposition}

\begin{proof}
When $\mu$ is finitely supported this is clearly equivalent to Definition~\ref{def:ce}.  The general case is part $(1)$ of Corollary $2.14$ in \cite{spo:cecgcc} with the present utilities substituted in.
\end{proof}

Note that these conditions are homogeneous (that is, invariant under positive scaling) in $\mu$.  The only condition on $\mu$ that is not homogeneous is the probability measure condition $\mu(I\times I) = 1$.  We will often ignore this condition to avoid having to normalize every expression, referring to a measure $\mu\in\Delta^*(C_X\times C_Y)$ satisfying the conditions of the proposition as a correlated equilibrium.

\begin{definition}When we need to distinguish these notions, we will refer to a measure $\mu\in\Delta^*(C_X\times C_Y)$ satisfying the conditions of Proposition~\ref{prop:correqchar1} as a \textbf{homogeneous correlated equilibrium} and a measure $\mu\in\Delta(C_X\times C_Y)$ satisfying the conditions as a \textbf{proper correlated equilibrium}.  In the context of homogeneous correlated equilibria the term \textbf{extreme} will refer to extreme rays; for proper correlated equilibria it will refer to extreme points.
\end{definition}

When $\mu\neq 0$ is a homogenous correlated equilibrium, $\frac{1}{\mu(I\times I)}\mu$ is a proper correlated equilibrium.  The set of homogenous correlated equilibria is a convex cone.  The extreme rays of this cone are exactly those measures which are positive multiples of the extreme points of the set of proper correlated equilibria.

The following proposition characterizes correlated equilibria of games satisfying Assumption~\ref{assump:mp} and is analogous to Proposition~\ref{prop:nashzeromean} for Nash equilibria.  Note how the Nash equilibrium measures were characterized in terms of their moments but the correlated equilibria are not.  Whereas the Nash equilibria are pairs of mixed strategies with zero mean for each player, condition (\ref{item:mpxint}) of this proposition says that the correlated equilibria are joint distributions such that regardless of each player's own recommendation, the conditional mean of his opponent's recommended strategy is zero.

\begin{proposition}
\label{prop:mpcorreqchar}
For a game satisfying Assumption~\ref{assump:mp} and a measure $\mu\in\Delta^*(C_X\times C_Y)$ such that $xy\neq 0$ $\mu$-a.e., the following are equivalent:
\begin{enumerate}
\item \label{item:mpcorreq}$\mu$ is a correlated equilibrium;
\item \label{item:mpxyint}
\begin{equation*}
\kappa_x(A) := \int_{A\times I} xy\,d\mu(x,y)\text{\ \ \ and\ \ \ }\kappa_y(A) := \int_{I\times A}xy\,d\mu(x,y)
\end{equation*}
are both the zero measure, i.e., equal zero for all measurable $A\subseteq I$;
\item \label{item:mpxint}
\begin{equation*}
\lambda_x(A) := \int_{A\times I} y\,d\mu(x,y)\text{\ \ \ and\ \ \ }\lambda_y(A) := \int_{I\times A} x\,d\mu(x,y)
\end{equation*}
are both the zero measure.
\end{enumerate}
\end{proposition}

\begin{proof}
(\ref{item:mpcorreq} $\Rightarrow$~\ref{item:mpxyint}) We will consider only $\kappa_x$; $\kappa_y$ is similar.  The conditions of Proposition~\ref{prop:correqchar1} with $A = I$ imply that
\begin{equation*}
x'\int_{I\times I}y\,d\mu(x,y) \leq \int_{I\times I}xy\,d\mu(x,y)\leq y'\int_{I\times I}x\,d\mu(x,y)
\end{equation*}
for all $x'\in C_X, y'\in C_Y$.  By assumption it is possible to choose $x'$ and $y'$ either positive or negative, so $\int_{I\times I}xy\,d\mu(x,y) = 0$.  A similar argument with any $A$ implies that $\int_{A\times I}xy\,d\mu(x,y)\geq 0$.  Therefore we have
\begin{equation*}
0 = \int_{I\times I}xy\,d\mu(x,y) = \int_{A\times I}xy\,d\mu(x,y) + \int_{(I\setminus A)\times I}xy\,d\mu(x,y)\geq 0 + 0 = 0
\end{equation*}
for all $A$, so the inequality must be tight and we get $\int_{A\times I} xy\,d\mu(x,y) = 0$ for all $A$.

(\ref{item:mpxyint} $\Leftrightarrow$~\ref{item:mpxint}) By definition $d\kappa_x = x\,d\lambda_x$ and by assumption $\lambda_x(0) = 0$.  If one of these measures is zero then so is the other, and respectively with $y$ in place of $x$.

(\ref{item:mpxyint} \& \ref{item:mpxint} $\Rightarrow$~\ref{item:mpcorreq}) The integrals in Proposition~\ref{prop:correqchar1} vanish.
%
\end{proof}

\begin{proposition}
\label{prop:mpextfinsuppcorr}
Fix a game satisfying Assumption~\ref{assump:mp}.  Let $k > 0$ be even and $x_1,\ldots, x_{2k}$ and $y_1,\ldots,y_{2k}$ be such that:
\begin{enumerate}
\item $x_i\in C_X$ and $y_i\in C_Y$ are all nonzero;
\item the sequence $x_1,x_3,\ldots, x_{2k-1}$ has distinct elements and alternates in sign;
\item the sequence $y_1,y_3,\ldots,y_{2k-1}$ has distinct elements and alternates in sign;
\item $x_{2i} = x_{2i-1}$ and $y_{2i} = y_{2i+1}$ for all $i$ when subscripts are interpreted $\mod 2k$.
\end{enumerate}
Then $\mu = \sum_{i=1}^{2k} \frac{1}{\lvert x_i y_i\rvert}\delta_{(x_i,y_i)}$ is an extreme correlated equilibrium.
\end{proposition}

\begin{proof}
To show that $\mu$ is a correlated equilibrium define $d\kappa(x,y) = xy\,d\mu(x,y)$.  Then $\kappa = \sum_{i=1}^{2k} \sign(x_i)\sign(y_i)\delta_{(x_i,y_i)}$.  Defining the projection $\kappa_x$ as in Proposition~\ref{prop:mpcorreqchar}, we have
\begin{equation*}
\begin{split}
\kappa_x & = \sum_{i=1}^{2k} \sign(x_i)\sign(y_i)\delta_{x_i} = \sum_{i=1}^k \sign(x_{2i})\left(\sign(y_{2i}) + \sign(y_{2i-1})\right)\delta_{x_{2i}} \\ &= \sum_{i=1}^k\sign(x_{2i})(0)\delta_{x_{2i}} = 0,
\end{split}
\end{equation*}
because $x_{2i} = x_{2i-1}$ and $y_{2i}$ differs in sign from $y_{2i-1}$ by assumption.  The same argument shows that $\kappa_y = 0$, so $\mu$ is a correlated equilibrium.

To see that $\mu$ is extreme, suppose $\mu = \mu' + \mu''$ where $\mu'$ and $\mu''$ are correlated equilibria.  Clearly $\mu' = \sum_{i=1}^{2k}\alpha_i\delta_{(x_i,y_i)}$ for some $\alpha_i\geq 0$.  Define $d\kappa' = xy\,d\mu'(x,y)$, so $\kappa' = \sum_{i=1}^{2k}\alpha_ix_iy_i\delta_{(x_i,y_i)}$.  By assumption
\begin{equation*}
\kappa'_x = \sum_{i=1}^{k}x_{2i}\left(\alpha_{2i-1}y_{2i-1} + \alpha_{2i}y_{2i}\right)\delta_{x_{2i}}
\end{equation*}
is the zero measure.  Since the $x_{2i}$ are distinct and nonzero we must have $\alpha_{2i-1}y_{2i-1} + \alpha_{2i}y_{2i} = 0$ for all $i$.  Similarly since $\kappa'_y = 0$ we have $\alpha_{2i+1}x_{2i+1} + \alpha_{2i}x_{2i} = 0$ for all $i$ (with subscripts interpreted $\mod 2k$).

The $x_i$ and $y_i$ are all nonzero, so fixing one $\alpha_i$ fixes all the others by these equations.  That is to say, these equations have a unique solution up to multiplication by a scalar, so $\mu'$ is a positive scalar multiple of $\mu$.  But the splitting $\mu = \mu' + \mu''$ was arbitrary, so $\mu$ is extreme.
\end{proof}

An argument along the lines of the proof of Proposition~\ref{prop:mpextfinsuppcorr} shows that any finitely supported correlated equilibrium $\mu$ whose support does not contain any points with $x=0$ or $y=0$ can be written as $\mu = \mu' + \mu''$ where $\mu'\neq 0$ is a correlated equilibrium and $\mu''\neq 0$ is a correlated equilibrium of the form studied in Proposition~\ref{prop:mpextfinsuppcorr}.  Therefore a finitely supported $\mu$ cannot be extreme unless it is of this form.

\usesavedcounter{exfin}
\begin{example}[cont'd]
For some examples of the supports of extreme correlated equilibria of games of this type, see Figures~\ref{fig:exfin1} and~\ref{fig:exfin2}.

\begin{figure}
\begin{center}
 \includegraphics[width = 0.4\textwidth]{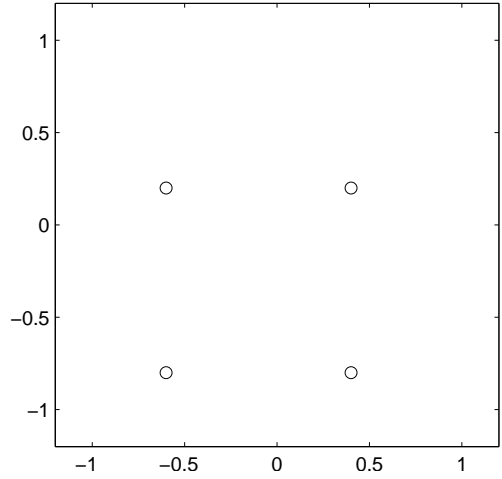}
\caption{The support of an extreme correlated equilibrium.  In the notation of Proposition~\ref{prop:mpextfinsuppcorr}, $k=2$, $x_1 = 0.4$, $x_3 = -0.6$, $y_1 = 0.2$, and $y_3 = -0.8$.}
\label{fig:exfin1}
\end{center}
\end{figure}

\begin{figure}
\begin{center}
 \includegraphics[width=0.4\textwidth]{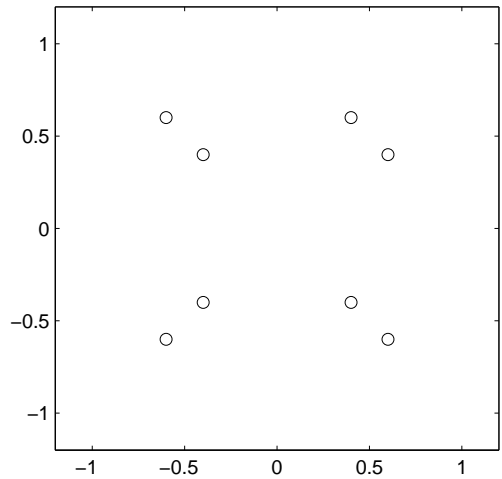}
\caption{The support of another extreme correlated equilibrium.  In the notation of Proposition~\ref{prop:mpextfinsuppcorr}, $k=4$, $x_1 = 0.4$, $x_3 = -0.4$, $x_5 = 0.6$, $x_7 = -0.6$, $y_1 = 0.6$, $y_3 = -0.4$, $y_5 = 0.4$, and $y_7 = -0.6$.}
\label{fig:exfin2}
\end{center}
\end{figure}

To count the number of extreme correlated equilibria of this game we must count the number of essentially different sequences of $x_i$ and $y_i$ of the type mentioned in Proposition~\ref{prop:mpextfinsuppcorr}.  Fix $k$ and let $k = 2r$ where $1\leq r \leq n$.  Note that cyclically shifting the sequences of $x_i$'s and $y_i$'s by two does not change $\mu$, nor does reversing the sequence.  Therefore we can assume without loss of generality that $x_1,y_1 > 0$.  We then have $n$ possible choices for $x_1,y_1,x_3,\text{ and }y_3$, $n-1$ possible choices for $x_5,x_7,y_5,\text{ and }y_7$, etc., for a total of $\left(\frac{n!}{(n-r)!}\right)^4$ possible choices of the $x_i$ and $y_i$.  These will always be essentially different (i.e., give rise to different $\mu$) unless we cyclically permute the sequences of $x_i$ and $y_i$ by some multiple of four, in which case the resulting sequence is essentially the same.  The number of such cyclic permutations is $r$.  Therefore the total number of extreme correlated equilibria is
\begin{equation*}
e(n) = \sum_{r = 1}^n \frac{1}{r}\left(\frac{n!}{(n-r)!}\right)^4.
\end{equation*}

We will see that $e(n) = \Theta\left(\frac{1}{n}(n!)^4\right)$.  That is to say, $e(n)$ is asymptotically upper and lower bounded by a constant times $\frac{1}{n}(n!)^4$.  The expression $\frac{1}{n}(n!)^4$ is just the final term in the summation for $e(n)$, so the lower bound is clear.  Define
\begin{equation*}
f(n) = \frac{e(n)}{\frac{1}{n}(n!)^4} = \sum_{s = 0}^{n-1} \frac{n}{n-s}\cdot\frac{1}{(s!)^4}.
\end{equation*}
Then $f(n)\geq 1$ for all $n$.  We will now show that $f(n)$ is also bounded above.  Intuitively this is not surprising since the terms in the summation for $f(n)$ die off extremely quickly as $s$ grows.

For all $1\leq s < n - 1$ we have that the ratio of term $s+1$ in the summation to term $s$ is:
\begin{equation*}
\frac{\frac{n}{n-s-1}\cdot\frac{1}{((s+1)!)^4}}{\frac{n}{n-s}\cdot\frac{1}{(s!)^4}} = \frac{n-s}{n-s-1}\cdot\frac{1}{(s+1)^4} \leq \frac{1}{8},
\end{equation*}
so for $n > 1$ we can bound the sum by a geometric series:
\begin{equation*}
f(n) - 1 = \sum_{s=1}^{n-1} \frac{n}{n-s}\cdot\frac{1}{(s!)^4} \leq \frac{n}{n-1}\sum_{t=0}^{\infty}\frac{1}{8^t} = \frac{8n}{7(n-1)}\leq \frac{16}{7}.
\end{equation*}

Therefore $1\leq f(n)\leq\frac{23}{7}$ for all $n$, so $e(n) = \Theta\left(\frac{1}{n}(n!)^4\right)$ as claimed.  Comparing this to the results of the previous section in which we saw that the number of extreme Nash equilibria of this game is $n^4$, we see that in this case there is a super-exponential separation between the number of extreme Nash and the number of extreme correlated equilibria.  This implies, for example, that computing all extreme correlated equilibria is not an efficient method for computing all extreme Nash equilibria, even though all extreme Nash equilibria are extreme correlated equilibria and recognizing whether an extreme correlated equilibrium is an extreme Nash equilibrium is easy.  There are simply too many extreme correlated equilibria.
\end{example}
\restorecounter

Next we will prove a more abstract version of Proposition~\ref{prop:mpextfinsuppcorr} which includes certain extreme points which are not finitely supported.  Before doing so we need a brief digression to ergodic theory.  The first definition is the standard definition of compatibility between a measure and a transformation on a space.  The second definition expresses one notion of what it means for a transformation to ``mix up'' a space -- in this case that the space cannot be partitioned into two sets of positive measure which do not interact under the transformation.  Then we state the main ergodic theorem and a corollary which we will apply to exhibit extreme correlated equilibria of games satisfying Assumption~\ref{assump:mp}.

\begin{definition}
Given a measure $\mu\in\Delta^*(S)$ on a space $S$, a measurable function $g: S\rightarrow S$ is called \textbf{($\mu$-)measure preserving} if $\mu(g^{-1}(A)) = \mu(A)$ for all measurable $A\subseteq S$.  Note that if $g$ is invertible (in the measure theoretic sense that an almost everywhere inverse exists), then this is equivalent to the condition that $\mu(g(A)) = \mu(A)$ for all $A$.
\end{definition}

\begin{definition}
Given a measure $\mu\in\Delta^*(S)$, a $\mu$-measure preserving transformation $g$ is called \textbf{ergodic} if $\mu(A\bigtriangleup g^{-1}(A))=0$ implies $\mu(A) = 0$ or $\mu(A) = \mu(S)$, where $A\bigtriangleup B$ denotes the symmetric difference $(A\setminus B) \cup (B\setminus A)$.
\end{definition}

\begin{example}
Fix a finite set $S$ and a function $g: S\rightarrow S$.  Let $\mu$ be counting measure on $S$.  Then $g$ is measure preserving if and only if it is a permutation.  In this case a set $T$ satisfies $\mu(g^{-1}(T)\bigtriangleup T) = 0$ if and only if $g^{-1}(T) = T$ if and only if $T$ is a union of cycles of $g$.  Therefore $g$ is ergodic if and only if it consists of a single cycle.
\end{example}

\begin{example}
\label{ex:ergrot}
Fix $\alpha\in\Rm$.  Let $S = [0,1)$ and let $\mu$ be Lebesgue measure on $S$.  Define $g: S\rightarrow S$ by $g(x) = (x+\alpha)\mod 1 = (x+\alpha)-\lfloor x+\alpha\rfloor$.  Then $g$ is $\mu$-measure preserving because Lebesgue measure is translation invariant.  It can be shown that $g$ is ergodic if and only if $\alpha$ is irrational.  For a proof and more examples, see \cite{silva:iet}.
\end{example}

The following is one of the core theorems of ergodic theory.  We will only use it to prove the corollary which follows, so it need not be read in detail.  The proof can be found in any text on ergodic theory, e.g.\ \cite{silva:iet}.

\begin{theorem}[Birkhoff's ergodic theorem]
Fix a probability measure $\mu$ and a $\mu$-measure preserving transformation $g$.  Then for any $f\in \mathcal{L}^1(\mu)$:
\begin{itemize}
\item $\tilde{f}(x) = \lim_{n\rightarrow\infty} \frac{1}{n}\sum_{k=0}^{n-1}f(g^k(x))$ exists $\mu$-almost everywhere,
\item $\tilde{f}\in\mathcal{L}^1(\mu)$,
\item $\int \tilde{f}\,d\mu = \int f\,d\mu$,
\item $\tilde{f}(g(x)) = \tilde{f}(x)$ $\mu$-almost everywhere, and
\item if $g$ is ergodic then $\tilde{f}(x) = \int f\,d\mu$ $\mu$-almost everywhere.
\end{itemize}
\end{theorem}

\begin{corollary}
\label{cor:uniquemeasure}
Suppose $\mu$ and $\nu$ are probability measures such that $\nu$ is absolutely continuous with respect to $\mu$.  If a transformation $g$ preserves both $\mu$ and $\nu$ and $g$ is ergodic with respect to $\mu$, then $\nu = \mu$.
\end{corollary}

\begin{proof}Fix any measurable set $A$.  Let $f$ be the indicator function for $A$, i.e.\ the function equal to unity on $A$ and zero elsewhere.  Applying Birkhoff's ergodic theorem to $f$ and $\mu$ yields $\tilde{f}(x) = \mu(A)$ $\mu$-almost everywhere.  Since $\nu$ is absolutely continuous with respect to $\mu$, $\tilde{f}(x) = \mu(A)$ $\nu$-almost everywhere also.  If we now apply Birkhoff's ergodic theorem to $\nu$ we get:
\begin{equation*}
\nu(A) = \int f \,d\nu = \int \tilde{f}\,d\nu = \int \mu(A)\,d\nu = \mu(A).\qedhere
\end{equation*}
\end{proof}



We now construct a family of extreme correlated equilibria.

\begin{proposition}
\label{prop:mpextcorr}
Fix measures $\nu_1, \nu_2, \nu_3$, and $\nu_4\in\Delta^*((0,1])$ and maps $f_i: (0,1]\rightarrow (0,1]$ such that $\nu_{i+1} = \nu_i\circ f_i^{-1}$ (interpreting subscripts $\mod 4$).  The portion of the measure $\mu$ in the $i^{\text{th}}$ quadrant of $I\times I$ will be constructed in terms of $f_i$ and $\nu_i$.  Define $j_i:(0,1]\rightarrow I\times I$ by $j_1(x) = (x,f_1(x))$, $j_2(x) = (-f_2(x),x)$, $j_3(x) = (-x,-f_3(x))$, and $j_4(x) = (f_4(x),-x)$.  Let $\abskap = \sum_{i=1}^4 \nu_i\circ j_i^{-1}$.  If Assumption~\ref{assump:mp} is satisfied, $\supp\abskap\subseteq C_X\times C_Y$, and $\frac{1}{\lvert xy\rvert}\in\mathcal{L}^1(\abskap)$ then $d\mu = \frac{1}{\lvert xy\rvert}\,d\abskap$ is a correlated equilibrium.

By assumption 
$f_4\circ f_3\circ f_2\circ f_1:(0,1]\rightarrow (0,1]$ is $\nu_1$-measure preserving.  If it is also ergodic with respect to $\nu_1$, then $\mu$ is extreme.
\end{proposition}

\begin{proof}
First we must show that $\mu$ is a correlated equilibrium.  It is a finite measure by the assumption $\frac{1}{\lvert xy\rvert}\in\mathcal{L}^1(\abskap)$ and $xy\neq 0$ $\mu$-a.e.\ by definition.  Define $g: I\times I \rightarrow I\times I$ as follows.
\begin{equation*}
g(x,y) = \begin{cases}
j_1(x) & \text{if }x>0, y<0 \\
j_2(y) & \text{if }x>0, y>0 \\
j_3(-x) & \text{if }x<0,y>0 \\
j_4(-y) & \text{if }x<0,y<0 \\
\text{arbitrary} & \text{otherwise}
\end{cases}
\end{equation*}
The function $g$ is $\abskap$-measure preserving.  To see this fix any measurable set $B\subseteq (0,1]\times(0,1]$.  Let $A = j_1^{-1}(B)$.  Then $\abskap(B) = \abskap(A\times (0,1]) = \nu_1(A)$ by definition of $\abskap$.  But $g^{-1}(B) = g^{-1}(A\times(0,1]) = A\times [-1,0)$, so 
\begin{equation*}
\abskap(g^{-1}(B)) = \abskap(A\times [-1,0)) = \nu_4(j_4^{-1}(A\times [-1,0))) = \nu_4(f_4^{-1}(A)) = \nu_1(A) = \abskap(B).
\end{equation*}
Therefore $g$ is measure preserving for subsets of $(0,1]\times (0,1]$.  The arguments for the other quadrants are similar and since $g$ maps each quadrant into a different quadrant, $g$ is measure preserving on its entire domain.

Define the signed measure $\kappa$ by $d\kappa = xy \,d\mu = \sign(x)\sign(y)\,d\abskap$.  We have seen that $\abskap(A\times (0,1]) = \abskap(A\times [-1,0))$, so $\kappa(A\times (0,1]) = -\kappa(A\times [-1,0))$.  Since $\kappa(A\times \{0\}) = 0$, we have $\kappa(A\times I) = 0$, or using the terminology of Proposition~\ref{prop:mpcorreqchar}, $\kappa_x(A) = 0$.  A similar argument implies $\kappa_x(A) = 0$ if $A\subseteq [-1,0)$.  Clearly $\kappa_x(0) = 0$ by definition of $\kappa_x$, so $\kappa_x$ is the zero measure.  In the same way we can show that $\kappa_y$ is the zero measure, so $\mu$ is a correlated equilibrium by Proposition~\ref{prop:mpcorreqchar}.

Now we will show via several steps that $\mu$ is extreme.  Write $\mu = \mu_1 + \mu_2$ where the $\mu_i$ are nonzero correlated equilibria.  Since these are all positive measures, the $\mu_i$ are absolutely continuous with respect to $\mu$.  
Define $d\abskapi = \lvert xy\rvert \,d\mu_i$ and $d\kappa_i = xy\,d\mu_i$.

Next we show that $g$ is $\abskapi$-measure preserving.  We will demonstrate this fact for $B\subseteq (0,1]\times (0,1]$.  As above, we define $A = j_1^{-1}(B)$.  Then $\abskapi(B) = \abskapi(A\times (0,1])$ since $(A\times (0,1]) \bigtriangleup B$ has $\abskap$ measure zero and $\abskapi$ is absolutely continuous with respect to $\abskap$.  Furthermore, $\abskapi(g^{-1}(B)) = \abskapi(A\times [-1,0))$.  But $\mu_i$ is a correlated equilibrium so $\kappa_i(A\times (0,1]) = -\kappa_i(A\times [-1,0))$.  Hence $\abskapi(g^{-1}(B)) = \abskapi(A\times [-1,0)) = \abskapi(A\times (0,1]) = \abskapi(B)$.  Again, the proof is the same for $B$ contained in other quadrants, so $g$ is $\abskapi$-measure preserving.  

For the second-to-last step we prove that $g$ is ergodic with respect to $\abskap$.  Suppose $B\subseteq I\times I$ is such that $\abskap(g^{-1}(B)\bigtriangleup B) = 0$.  Let $Q_i$ be the intersection of $B$ with the $i^{\text{th}}$ quadrant.  Then $\abskap(g^{-1}(Q_{i+1})\bigtriangleup Q_i) = 0$, so $\abskap(g^{-4}(Q_1) \bigtriangleup Q_1) = 0$.  Let $A = j_1^{-1}(Q_1)$.  Then $\abskap(g^{-4}(Q_1)\bigtriangleup Q_1) = \nu_1((f_4\circ f_3 \circ f_2\circ f_1)^{-1}(A) \bigtriangleup A) = 0$.  By assumption the map $f_4\circ f_3\circ f_2\circ f_1$ is ergodic, so $\nu_1(A) = 0$ or $\nu_1(A) = \nu_1((0,1]) = \abskap((0,1]\times(0,1])$.  Therefore $\abskap(Q_1) = \nu_1(A) = 0$ or $\abskap(Q_1) = \abskap((0,1]\times(0,1])$.  In either case since $g$ is $\abskap$-measure preserving we get $\abskap(Q_i) = \abskap(Q_1)$ for all $i$.  Therefore $\abskap(B) = 0$ or $\abskap(B) = \abskap(I\times I)$, so $g$ is ergodic with respect to $\abskap$.

Normalizing $\abskap$ and $\abskapi$ to be probability measures, we can apply Corollary~\ref{cor:uniquemeasure} to obtain $\abskapi = \frac{\abskapi(I\times I)}{\abskap(I\times I)}\abskap$.  By definition the set on which $\lvert xy \rvert$ is zero has $\mu$ measure zero.  Therefore
\begin{equation*}
d\mu_i = \frac{1}{\lvert xy \rvert}\,d\abskapi =   \frac{\abskapi(I\times I)}{\abskap(I\times I)}\frac{1}{\lvert xy \rvert}\,d\abskap =  \frac{\abskapi(I\times I)}{\abskap(I\times I)}\,d\mu,
\end{equation*}
so $\mu_i =  \frac{\abskapi(I\times I)}{\abskap(I\times I)} \mu$ and $\mu$ is extreme.
\end{proof}

Above we have constructed $\mu$ and $g$ so that $g$ maps the quadrants counter-clockwise -- quadrant $1$ to quadrant $2$, etc.  However, the same argument would go through if $g$ mapped the quadrants clockwise.

To view Proposition~\ref{prop:mpextfinsuppcorr} as a special case of Proposition~\ref{prop:mpextcorr}, let each $\nu_i$ be a uniform probability measure over a finite subset of $(0,1]$.  The function $g$ is defined by $g(x_i,y_i) = (x_{i+1},y_{i+1})$ and the $f_i$ are defined to be compatible with this.  The map $f_4\circ f_3\circ f_2\circ f_1$ is a permutation on the support of $\nu_1$, which is precisely the positive values of $x_i$.  By construction this permutation consists of a single cycle, hence it is ergodic.

\usesavedcounter{exinf}
\begin{example}[cont'd]
We can combine Example~\ref{ex:ergrot} and Proposition~\ref{prop:mpextcorr} to exhibit extreme points of the set of correlated equilibria for this game which are not finitely supported.  Let $0 < a < b < 1$.  Let $\nu_i$ be Lebesgue measure on $[a,b)$ for all $i$.  Fix $\alpha$ such that $\frac{\alpha}{b-a}$ is irrational.  Define $f_1: [a,b)\rightarrow [a,b)$ by $f(x) = (x-a+\alpha\mod (b-a)) + a$.  This is just an affinely scaled version of Example~\ref{ex:ergrot} so $f_1$ is $\nu_i$-measure preserving and ergodic.  Define $f_1$ on $(0,1]\setminus [a,b)$ arbitrarily, because that is a set of measure zero.  Let $f_2,f_3,f_4: (0,1]\rightarrow (0,1]$ be the identity.  These data satisfy all the assumptions of Proposition~\ref{prop:mpextcorr}.  In particular, since $0<a<b<1$, $xy$ is bounded away from zero on the support of $\abskap$.  Therefore $\frac{1}{\lvert xy \rvert}\in\mathcal{L}^1(\abskap)$.  Since $\nu_i$ is not finitely supported, $\mu$ is an extreme correlated equilibrium which is not finitely supported.  The support of $\mu$ is shown in Figure~\ref{fig:exinf} with parameters $a = 0.2$, $b = 0.8$, and $\alpha = \frac{1}{\sqrt{5}}$.
\restorecounter

\begin{figure}
\begin{center}
 \includegraphics[width = 0.4\textwidth]{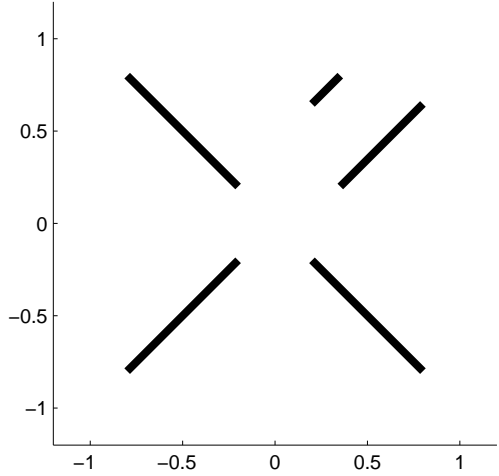}
\caption{The support set of an extreme correlated equilibrium which is not finitely supported.  Extremality of this equilibrium depends sensitively on the choices of endpoints for the line segments.  In this case there are segments connecting: $(0.2,-0.2)$ to $(0.8, -0.8)$; $(-0.2,-0.2)$ to $(-0.8,-0.8)$; $(-0.2,0.2)$ to $(-0.8,0.8)$; $\left(0.2,0.2+\frac{1}{\sqrt{5}}\right)$ to $\left(0.8-\frac{1}{\sqrt{5}},0.8\right)$; and $\left(0.8-\frac{1}{\sqrt{5}},0.2\right)$ to $\left(0.8,0.2+\frac{1}{\sqrt{5}}\right)$.}
\label{fig:exinf}
\end{center}
\end{figure}
\end{example}
\restorecounter

\begin{definition}Given a compact Hausdorff space $K$ we say that a set of measures $\mathcal{M}\subseteq\Delta^*(K)$ is \textbf{describable by moments} if there exists an integer $d$, bounded Borel measurable maps $g_1,\ldots,g_d: K\rightarrow\Rm$, and a set $M\subseteq\Rm^d$ such that a measure $\mu$ is in $\mathcal{M}$ if and only if $\left(\int g_1\,d\mu,\ldots,\int g_d\,d\mu\right)\in M$.
\end{definition}

The results of \cite{karlin:tig} show that the maximin and minimax strategy sets of a two-player zero-sum polynomial game can always be described by moments.  Introducing a similar notion for $n$-tuples of measures, the set of Nash equilibria can always be described by moments in any polynomial game \cite{sop:slrcg}.  However, combining this example with the following proposition we see that the set of correlated equilibria of a polynomial game cannot in general be described by moments.

This is important because the finite-dimensional representation in terms of moments is the primary tool for computing and characterizing Nash equilibria of polynomial games.  One is therefore naturally drawn to try to find such a representation for the set of correlated equilibria.  The example and this proposition show that no such representation exists in general.

\begin{proposition}
\label{prop:momentmeasure}
Let $\mathcal{M}\subseteq\Delta^*(K)$ be a set of measures describable by moments.  Then all extreme points of $\mathcal{M}$ have finite support and this support is uniformly bounded by $d$, where $d$ is the integer associated with the description of $\mathcal{M}$ by moments.
\end{proposition}

\begin{proof}
Let $g_1,\ldots,g_d: K\to\Rm$ be the maps describing $\mathcal{M}$.  Suppose there exists a measure $\mu\in\mathcal{M}$ which is extreme and supported on more than $d$ points, so we can partition the domain of $\mu$ into $d+1$ sets $B_1,\ldots,B_{d+1}$ of positive measure.  For $c = (c_1,\ldots, c_{d+1})\in\Rm^{d+1}_{\geq 0}$, define $\mu_c = \sum_{i = 1}^{d+1} c_i \mu\vert_{B_i}$.  The map $c\mapsto \mu_c$ is injective.  Define
\begin{equation*}
K = \left\{c\in\Rm^{d+1}_{\geq 0} \,\bigg\vert\, \int g_i\,d\mu_c=\int g_i\,d\mu\text{ for }i =1,\ldots, d\right\},
\end{equation*}
so $(1,1,\ldots, 1)\in K$.  Linearity of integration implies that the nonempty set $K$ is the intersection of an affine space of dimension at least one with the positive orthant.  By Carath\'{e}odory's theorem (or equivalently the statement that a feasible linear program has a basic feasible solution), the extreme points of $K$ each have at most $d$ nonzero entries \cite{bno:convex}.  Thus $(1,1,\ldots, 1)$ is not an extreme point of $K$, so we can write $(1,1,\ldots, 1) = \lambda c + (1 - \lambda) c'$ for $0<\lambda < 1$ and $(1,1,\ldots, 1)\neq c,c'\in K$.  Therefore $\mu = \mu_{(1,1\ldots, 1)} = \lambda \mu_c + (1-\lambda)\mu_{c'}$ is not extreme.
\end{proof}

\section{Future work}
These results leave several open questions.  If we define a moment map to be any map of the form $\pi\mapsto \left(\int f_1\,d\pi,\int f_2\,d\pi,\ldots,\int f_k\,d\pi\right)$ for bounded Borel measurable $f_i$, then we have shown that the set of correlated equilibria is not the inverse image of any set under any moment map.  On the other hand, since moment maps are linear and weak* continuous, we know that the image of the set of correlated equilibria under any moment map is convex and compact.

Supposing the utilities and the $f_i$ are polynomials, is there anything more we can say about this image?  In particular, is it semialgebraic (i.e.,\ describable in terms of finitely many polynomial inequalities)?  If so, can we compute these inequalities or a solution thereof efficiently for given utilities?  A sequence of easily computed outer bounds to this image is presented in \cite{spo:cecgcc}.  Can we compute nonempty inner bounds?

\section*{Acknowledgements}
The first author would like to thank Prof. Cesar E. Silva for many discussions about ergodic theory, and in particular for the simple proof of Corollary~\ref{cor:uniquemeasure} using Birkhoff's ergodic theorem.

\bibliographystyle{plain}
\bibliography{../../references}
\end{document}